\documentclass[letterpaper,11pt]{article}

\usepackage[margin=1.5in]{geometry}
\usepackage[colorlinks]{hyperref}
\hypersetup{linkcolor=blue,
citecolor=blue,
}
\usepackage[affil-it]{authblk}
\usepackage{amsthm}
\theoremstyle{plain}
\newtheorem{theorem}{Theorem}[section]

\newtheorem{definition}[theorem]{Definition}

\theoremstyle{definition}

\usepackage[T1]{fontenc}
\usepackage[dvipsnames,svgnames]{xcolor}
\usepackage{subfig}
\usepackage{adjustbox}
\usepackage{amsmath,amssymb}
\usepackage{stmaryrd} 

\usepackage{graphicx}
\newcommand{\QQ}{\ensuremath{\mathbf Q}}
\newcommand{\RR}{\ensuremath{\mathbf R}}
\newcommand{\R}{\ensuremath{\mathbb R}}
\newcommand{\PP}{\ensuremath{\mathbf{P}}}
\newcommand{\cV}{\ensuremath{\mathcal{V}}}
\newcommand{\W}{\ensuremath{\mathcal{W}}}
\newcommand{\cT}{\ensuremath{\mathcal{T}}}
\def\argmin{\mathop{\rm argmin}\limits}%
\def\argsup{\mathop{\rm argsup}\limits}%
\newcommand{\Exp}{\ensuremath{\text{Exp}}}
\newcommand{\Log}{\ensuremath{\text{Log}}}

\title{Monge-Kantorovich quantiles and ranks for image data}
\date{}

\begin{document}

\author{Gauthier Thurin\thanks{Corresponding author: \texttt{gauthier-louis.thurin@math.u-bordeaux.fr}}}

\affil{Institut de Math\'ematiques de Bordeaux,
Universit\'e de Bordeaux, CNRS UMR 5251 \\ 
351 Cours de la Lib\'eration, 33400 Talence cedex, France}

\maketitle

\vspace{-.5cm}

\begin{abstract}
This paper defines quantiles, ranks and statistical depths for image data by leveraging ideas from measure transportation. 
The first step is to embed a distribution of images in a tangent space, with the framework of linear optimal transport. 
Therein, Monge-Kantorovich quantiles are shown to provide a meaningful ordering of image data, with outward images having unusual shapes. 
Numerical experiments showcase the relevance of the proposed procedure, for descriptive analysis, outlier detection or statistical testing. 

\end{abstract}

\textbf{Keywords:} Monge-Kantorovich statistical depth ; Linear Optimal Transport.

\section{Introduction}

\subsection{Ordering multivariate data}

The notion of quantiles is fundamental to characterize and describe a probability distribution on the real line. 
At the core of its definition lies the left-to-right ordering of $\R$. 
Unfortunately, this does not canonically extend in higher dimension, because of the numerous ways to sort vectors within a point cloud \cite{chaudhuri1996geometric,chernozhukov2015mongekantorovich}. 
Of course, this issue becomes even more complicated for image data, due to the inherent non-linear structure. 
Hence, this paper
aims to provide a meaningful ordering between typical and atypical images, as exemplified in Figure \ref{Fig_extremes_digits_intro}. 

\begin{figure}[h]
\centering
\includegraphics[width=0.48\textwidth]{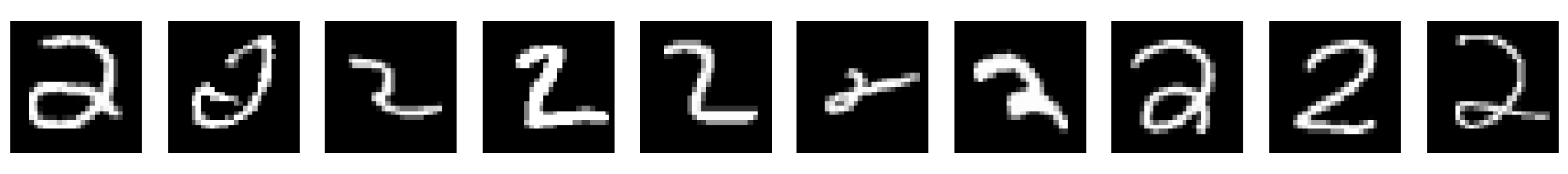}
\quad
\includegraphics[width=0.48\textwidth]{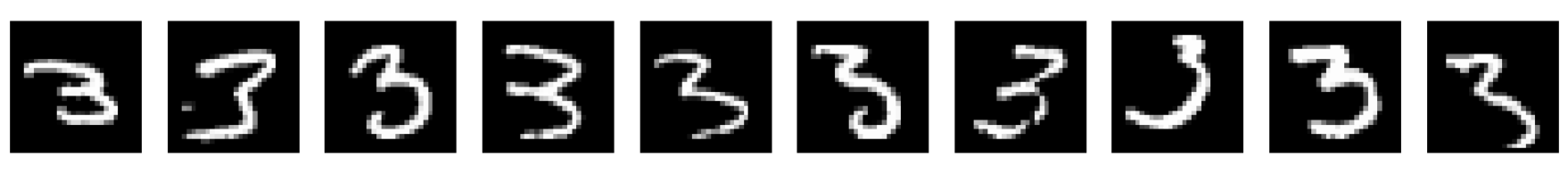}

\includegraphics[width=0.48\textwidth]{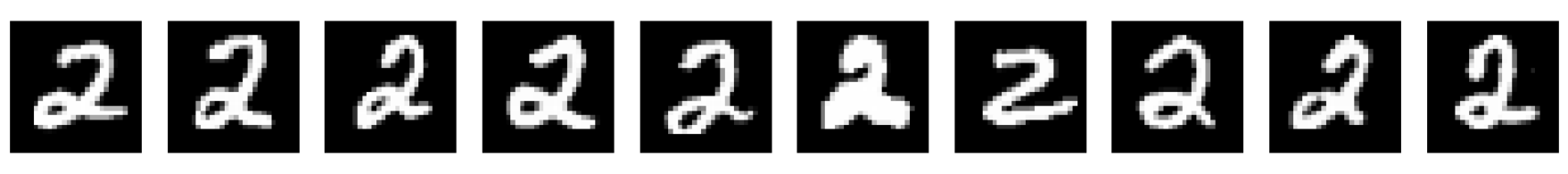}
\quad
\includegraphics[width=0.48\textwidth]{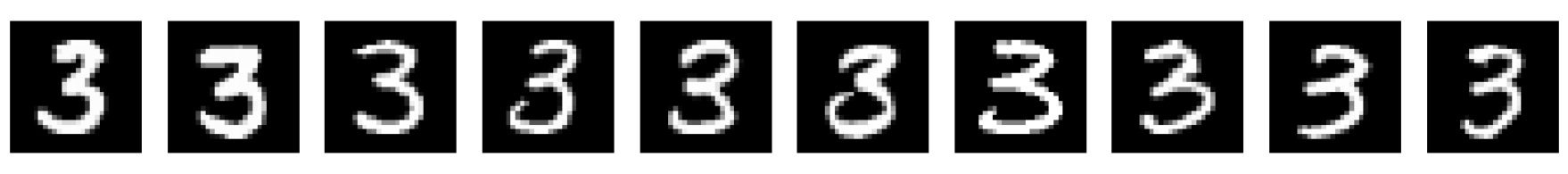}

\caption{Extreme (first row) and central (second row) images for two distributions.}
\label{Fig_extremes_digits_intro}
\end{figure}

As a general rule, one is willing to associate to each observation a measure of centrality with respect to some dataset, in view of descriptive analysis \cite{zuo2000general}, outlier detection \cite{chen2008outlier}, or statistical testing \cite{ghosal2022multivariate}, to name but a few.  
Most traditional concepts to order multivariate data include the Mahalanobis distance and the halfspace depth \cite{tukey75}, under Gaussian or convex assumptions, respectively. 
Other approaches that overcome these assumptions have gained interest among recent years, namely
spatial \cite{chaudhuri1996geometric} and Monge-Kantorovich (MK) quantiles \cite{chernozhukov2015mongekantorovich}. 
While both notions share a number of appealing features, MK quantiles better adapt to the shape of the underlying distribution \cite{girardhal00865767,hallinmultivariate}.
We refer to \cite{hallin2022measure} for an overview of applications of MK quantiles, depth, ranks and signs, together with insights about their desirable properties. 
Although MK quantiles have been defined on nonlinear manifolds \cite{hallin2024quantiles,hallin2024nonparametric}, their extension for image data remains, to the best of our knowledge, an open question.

\subsection{Images as histograms}

\paragraph{The Wasserstein space.}
The first step of our analysis is the choice of a metric between images.  
Up to normalization, images can be seen as histograms (or densities) on a set of pixels $\Omega$. 
Hence, each image $I$ is associated with a probability measure $d\rho(x) = I(x) dx$ in the set of square integrable distributions $\mathcal{P}_2(\Omega)$. 
This brings naturally the Wasserstein distance $\W_2$ that encodes information related with human perception \cite{peyre2019computational}. 
Thus, our goal amounts to define quantiles in the Wasserstein space $(\mathcal{P}_2,\W_2)$.
In order to \textit{linearize} this space while preserving the Wasserstein geometry, we resort to the Linear Optimal Transport (LOT) framework \cite{wang2013linear} that comes with several practical and theoretical advantages.

\paragraph{Linear Optimal Transport.}

LOT is a promising framework for distributional data that leverages the structure of the Wasserstein distance \cite{wang2013linear}.
Its main virtue is to embed distributions in a Hilbert space, allowing to apply Euclidean methods (e.g. PCA) quite simply, as evidenced by existing works 
\cite{basu2014detecting,cazelles2018geodesic,gachon2024low,khurana2023supervised,kolouri2015transport,kolouri2016continuous,martin2024data,moosmuller2023linear,park2018representing}.
Its construction comes from a Riemannian interpretation of $(\mathcal{P}_2(\Omega),\W_2)$ where, at any distribution $\rho \in \mathcal{P}_2(\Omega)$, a
 tangent space is given by a subset of $L^2(\Omega,\rho)$, the set of square integrable mappings \cite{ambrosio2008gradient}. 
 Therefore, embeddings in such tangent spaces must locally approximate the structure of $(\mathcal{P}_2(\Omega),\W_2)$. 

%

Transforming image data in a Hilbert space is a natural idea, and several methods exist for this task,  as in
\cite{Petersen2016} for univariate densities. 
However, LOT is especially appealing in view of exploratory analysis. 
Unlike other popular feature spaces (Fourier or Wavelets), it is not obtained via a linear operator \cite{martin2024data}. 
As opposed to deep-learning based features, the Euclidean distance between embeddings is an upper bound of the Wasserstein distance between images.   
Even more, and one can show bi-Hölder equivalence under some assumptions, 
\cite{delalande2023quantitative}. 

\subsection{Generalized quantiles.}

Other concepts of quantiles or statistical depth exist for complex structures, for functional data \cite{GijbelsNagy2017} or in metric spaces \cite{liu2022quantiles}. 
Most related to us is a very recent extension of the spatial depth for distributional data \cite{bachoc2024wasserstein}, that captures the geometry of $(\mathcal{P}_2(\Omega),\W_2)$.  
This valuable approach differs from our proposal in that it does not come with a concept of quantiles and it requires more computation.
Indeed, starting from $n$ images (or distributions) $I_1,\cdots,I_n$, it is required to compute $n$ OT problems to obtain the depth (or centrality) of any new image $I_{\rm new}$. 
In contrast, in our approach, the main computational cost concerns the embedding of $I_1,\cdots,I_n$, which requires to solve $n$ OT problems \textit{once for all}. 
After that, when facing a new image $I_{\rm new}$, a single OT problem is required to embed it in the LOT space and obtain its MK depth. 
In addition, MK quantiles and ranks benefit from specific attractive properties compared to spatial counterparts \cite{hallinmultivariate}, as for statistical tests \cite{deb2023multivariate} which are explored below. 

\subsection{Contributions and outline}

Our main contribution is to define MK quantiles for image data equipped with the Wasserstein distance. 
This comes with the associated notions of ranks and statistical depth \cite{chernozhukov2015mongekantorovich}. 
To do so, we make use of the LOT framework \cite{wang2013linear} to embed images in tangent spaces of $(\mathcal{P}_2,\W_2)$. 
Therein, we resort to point clouds in a high-dimensional linear space, where Euclidean definitions can be considered.  
This approach is illustrated in Figure \ref{Fig_illustration}.

\begin{figure}[h]
\centering
\includegraphics[width=0.95\textwidth]{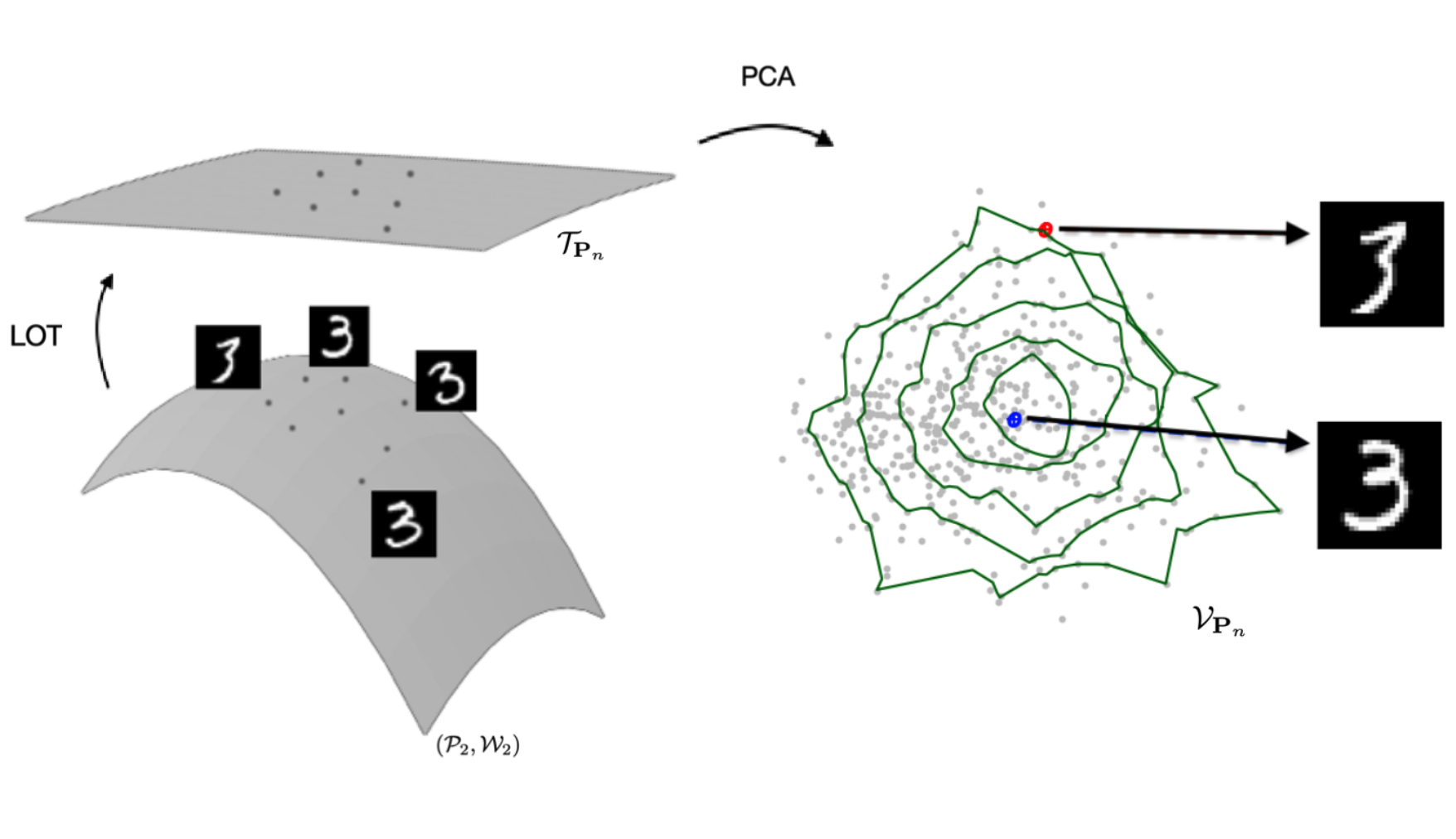}

\caption{MK quantiles in Log-PCA space, sorting images from common to atypical.}
\label{Fig_illustration}
\end{figure}

In Section \ref{MKquantiesinLOTspace}, our formalism and definitions are introduced. 
For the sake of parcimony, principal components analysis (PCA) is performed before considering MK quantiles. 
The potential loss of information is taken into account by a complementary statistical depth. 
In Section \ref{TheoreticalResults}, the consistency of our quantile function is studied. 
In Section \ref{NumExp}, several applications are pursued, 
from the very basis of descriptive analysis to outlier detection and statistical testing. 
For the sake of reproducibility, the codes used in the numerical experiments are made available at \url{https://github.com/gauthierthurin/QuantilesForImages}.


\section{MK quantiles in Log-PCA space}\label{MKquantiesinLOTspace}

\subsection{LOT embedding.} 

An image $I$ is viewed as a density distribution over a set $\Omega=\{ \omega_i\}_{i=1}^p\subset \R^2$ of $p$ pixels, that requires to normalize intensities of all pixels to sum to one. 
Equating densities to the corresponding measures in $\mathcal{P}_2(\Omega)$, a meaningful distance between two images $I$ and $J$ over $\Omega$ is the Wasserstein distance \cite{peyre2019computational}. 
It is defined by
\begin{equation}\label{MongeOT}
    \W_2^2(I,J) = \min_{T_\#I = J} \sum_{i=1}^p \Vert \omega_i - T(\omega_i)\Vert^2 I(\omega_i),
\end{equation}
where $T_\#I = J$ means that
$$
\forall j \in \{1,\cdots,p\}, \quad J(\omega_j) = \sum_{i: T(\omega_i) = \omega_j} I(\omega_i).
$$
Consider an empirical distribution $\PP_n$ over a set of images $\{I_1,\cdots,I_n\}$ in $\mathcal{P}_2(\Omega)$. 
The LOT embedding requires setting a template image $I_r$, typically chosen as an average \cite{park2018representing,wang2013linear}, such as the empirical Wasserstein barycenter.
Another alternative is $\overline{I}$ the pixel-wise empirical average, or,  
to ensure that it belongs to the underlying distribution, 
\begin{equation}\label{AvgImage}
    I_r = \argmin_{j \in \{1,\cdots,n\}} \Vert I_j - \overline{I} \Vert_2. 
\end{equation}
For $d\rho(x) = I_r(x) dx$ the associated measure, 
define the Logarithm map 
\begin{equation}\label{Logmap}
    \Log_{\PP_n}: I \mapsto T_I - \mbox{id},
\end{equation}
where $T_I$ is the Monge map that minimizes \eqref{MongeOT} between $I_r$ and $I$. 
Precisely, $\Log_{\PP_n}$ maps each image $I_k$ to $v_k = T_{I_k} - \mbox{id}$ inside the tangent space 
\begin{equation*}
\cT_{\PP_n} = \Big\{ v : \Omega \rightarrow \R^2 ; 
\sum_{i=1}^p \Vert v(\omega_i) \Vert^2 I_r(\omega_i) < +\infty \Big\}.
\end{equation*}
Note that, as opposed to the continuous case, the sum above is always finite for $p$ pixels, as soon as $v$ does not take infinite values. 
Our notations $\Log_{\PP_n}$ and $\cT_{\PP_n}$ emphasize that we consider one tangent space for each distribution ${\PP_n}$ of images. 
Crucially, in the Hilbert space $\cT_{\PP_n}$, the $2$-norm is an approximation of the Wasserstein distance, \cite{delalande2023quantitative}.
In practice, each $v\in \cT_{\PP_n}$ is assimilated to a vector of dimension $2p$. 
To obtain a lower-dimensional representation, it is common to use PCA \cite{cazelles2018geodesic,park2018representing,wang2013linear} as it preserves as much information as possible relatively to the $2$-norm in $\cT_{\PP_n}$. 
Doing so, we model vectors in $\cT_{\PP_n}$ as 
$
v = m + WX,
$
where $X$ belongs to a latent space $\cV_{\PP_n}$ of dimension $d<2p$, and $m$ and $W$ are obtained by PCA\footnote{Namely $m$ is the empirical mean of $\{v_i\}_{i=1}^n$ and $W = U\Lambda^{1/2}$ for $\Lambda$ diagonal with the $d$ largest eigenvalues of $\{v_i\}_{i=1}^n$ in decreasing order and $U$ contains the associated $d$ eigenvectors in columns.} on $\{ v_1,\cdots,v_n\}\subset \cT_{\PP_n}$. 
Hereafter, the final embedding vectors are elements of the Log-PCA space 
\begin{equation}
    \cV_{\PP_n} = \big\{ x \in \mathbb{R}^{d}: \exists v \in \cT_{\PP_n},\, x = W^{T}(v - m) \big\}.
\end{equation}

\subsection{Monge-Kantorovich quantiles.} 

Our definition of quantiles for a distribution ${\PP_n}$ of images corresponds to computing MK quantiles \cite{chernozhukov2015mongekantorovich,Hallin_AOS_2021} in the Log-PCA space $\cV_{\PP_n}$. 
Therein, a reference distribution $\mu$ must be chosen, that shall represent an ideal case.
Here, the latter is chosen as the spherical uniform, \cite{chernozhukov2015mongekantorovich,Hallin_AOS_2021} that is defined as the product $R\Phi$ between two independent random variables, uniformly drawn on $[0,1]$ and on the unit sphere $\{ z \in \mathbb{R}^d: \Vert z \Vert =1\}$, respectively.
With this choice,  
nested balls $\mathbb{B}(0,\alpha)$ of radius $\alpha$ are convenient reference quantile regions because $\mu(\mathbb{B}(0,\alpha))=\alpha$. 

Denote by $X_1,\cdots,X_n$ the representation in $\cV_{\PP_n}$ of the images $I_1,\cdots,I_n$. 
A dual version of Kantorovich OT problem writes
\begin{equation*}
    \psi_n = \argmin_{\varphi\in\Phi_0}  \int \varphi(u)d \mu(u)  + \frac{1}{n}\sum_{i=1}^n \varphi^*(X_i),
\end{equation*}
with $\varphi^*(x) = \sup_{u} \{\langle x,u \rangle - \varphi(u)\}$ the Legendre-Fenchel transform of $\varphi$ and $\Phi_0 = \{\varphi : (\varphi^*)^*=\varphi \}$.
Then, the MK rank function is given by 
\begin{equation}
    \RR_n(x) = \argsup_{ u: \Vert u \Vert \leq 1 } \big\{ \langle u,x \rangle - \psi_n(u) \big\},\label{defRn}
\end{equation}
and the MK quantile function is defined by 
\begin{equation}
    \QQ_n(u) = \argsup_{ x\in \{X_1,\cdots,X_n\} } \big\{ \langle u,x \rangle - \psi_n^*(x) \big\}.\label{defQn}
\end{equation}
This follows the convention of \cite{chernozhukov2015mongekantorovich,ghosal2022multivariate}, even if 
$\RR_n$ is called a distribution function in \cite{Hallin_AOS_2021}. 
Intuitively, $\RR_n$ associates to any $x$ a rank whose position describes the outlyingness of $x$, hence it plays the role of univariate ranks $\{1,\cdots,n\}$ but with added directional information. 
Denote by $\Exp_{\PP_N}$ the inverse mapping of $\Log_{\PP_N}$ in \eqref{Logmap}. 
We emphasize that, by construction, the image by $\Exp_{\PP_n} \big( \QQ_n (u)  \big)$ of $\mathbb{B}(0,\alpha)$ belongs to $(\mathcal{P}_2(\Omega),\W_2)$ and
has $\PP_n$-probability $\alpha$.
In other words, a proportion $\alpha$ of images lies in this region, making it an appropriate \textit{quantile region}. 
We now turn to the corresponding MK statistical depth \cite{chernozhukov2015mongekantorovich}.

\subsection{Inner and outer statistical depth.}

A statistical depth associates to each point $x$ a measure of centrality $D(x,\nu)\in \mathbb{R}$ with respect to a distribution $\nu$.
In the line of the LOT framework and for the sake of parcimony, we consider two different statistical depth functions. 
Put simply, the inner depth is dedicated to describe the variability within the Log-PCA space $\cV_{\PP_n}$, whereas the outer depth is dedicated to capture the outlyingness with respect to $\cV_{\PP_n}$. 
We shall now motivate this two-steps construction. 

The Monge-Kantorovich depth \cite{chernozhukov2015mongekantorovich} of $x \in \R^d$ with respect to ${\PP_n}$ is defined by the Tukey depth \cite{tukey75} of the rank $\RR_n(x)$ with respect to the reference $\mu$,
\begin{equation}\label{MKdepth}
D^{MK}(x,{\PP_n}) = D^{Tukey}\big( \RR_n(x),\mu\big).
\end{equation}
This only depends on $\Vert \RR_n(x) \Vert$, as shown by the explicit form of $D^{Tukey}( \cdot,\mu)$ in \cite{chernozhukov2015mongekantorovich} for $\mu$ the spherical uniform distribution.  
This definition straightforwardly applies within $\cV_{\PP_n}$, giving rise to the following. 
\begin{definition}[Inner depth, within one subspace]\label{innerdepth}
For a distribution ${\PP_n} = \frac{1}{n}\sum_{j=1}^n \delta_{I_j}$ of images in $(\mathcal{P}_2(\Omega),\W_2)$, an inner depth function is defined by 
$$
D^{i}(I,{\PP_n}) = D^{MK}(x,{\PP_n}),
$$
for $x = W^{T}(\Log_{\PP_n}(I) - m)$ the embedding of $I$ in the Log-PCA space $\cV_{\PP_n}$. 
\end{definition}

However, due to the projection step into $\cV_{\PP_n}$, the inner depth may show irrelevant results for out-of-distribution examples.
Namely, it is possible that an outlier in $\cT_{\PP_n}$ is projected by PCA in central regions of the distribution in $\cV_{\PP_n}$. 
In fact, the inner depth only captures the centrality within $\cV_{\PP_n} \subset \cT_{\PP_n}$.
The outer depth complements this by rendering outlyingness with respect to $\cV_{\PP_n}$.

\begin{definition}[Outer depth, between subspaces]\label{outerdepth}
For a distribution ${\PP_n} = \frac{1}{n}\sum_{j=1}^n \delta_{I_j}$ of images in $(\mathcal{P}_2(\Omega),\W_2)$, an outer depth function is defined by 
$$
D^{o}(I,{\PP_n}) = \Big( 1+ d\big(\Log_{\PP_n}(I),\cV_{\PP_n}\big) \Big)^{-1},
$$
with  
$$
d(\Log_{\PP_n}(I),\cV_{\PP_n}) = \Vert  \Log_{\PP_n}(I) - P_{\cV_{\PP_n}}\big(\Log_{\PP_n}(I)\big) \Vert_2,
$$
the distance between the embedding $\Log_{\PP_n}(I)$ and $\cV_{\PP_n}$, and where $P_{\cV_{\PP_n}}(\cdot)$ is the orthogonal projection from $\cT_{\PP_n}$ to $\cV_{\PP_n}$. 
\end{definition}

\section{Convergence of empirical counterparts}\label{TheoreticalResults}

For the sake of simplicity, the estimation of PCA is neglected throughout this section. 
In other words, we now suppose that $d=2p$ for $p$ pixels and that the PCA explains full variance.
Assume also that the reference image $I_r$ is learnt on $N$ images, when $N$ is fixed (although in practice $N=n$) and let $\PP_N$ be the associated empirical measure.

Consider now images $I_1,\cdots,I_n$ that are also drawn $i.i.d.$ from $\PP$ and let $\PP_n$ be the associated empirical measure. 
Define $\nu_n$ and $\nu$ the respective push-forward distributions $\Log_{\PP_N}\# \PP_n $ and $\Log_{\PP_N}\# \PP$. 
These are the embedded versions of $\PP_n$ and $\PP$ up to the choice of the tangent space $\cT_{\PP_N}$. 
Necessarily, the weak convergence $\PP_n \overset{w}{\longrightarrow} \PP$ holds, as well as $\nu_n \overset{w}{\longrightarrow} \nu$, with $N$ held fixed. 

Note that, in general, $T_\#\mu=\nu$ means that $T(U) \sim \nu$ if $U\sim\mu$. A continuous definition of MK quantiles reads as follows.
\begin{definition}[MK quantiles and ranks \cite{Hallin_AOS_2021}]\label{defMKquantiles}
    Let $\nu$ be a Lebesgue-absolutely continuous distribution over $\mathbb{R}^d$. 
    The MK quantile function of $\nu$ is the $a.e.$ unique mapping such that $\QQ_\#\mu=\nu$ and $\QQ=\nabla\psi$ for some $\psi:\mathbb{R}^d\rightarrow \mathbb{R}$.
    Similarly, the MK rank function of $\nu$ is $\RR(x) = \nabla \psi^*(x)$, for $\psi^*(x) = \sup_{u: \Vert u \Vert \leq 1} (\langle x,u\rangle - \psi(u) )$.
\end{definition}
Let $\QQ$ and $\RR$ be respectively the MK quantile and rank functions of $\nu$ in $\cT_{\PP_N}$ of dimension $d=2p$. 
Here, one might note that the distributions $\nu$ and $\nu_n$ (for all $n$) are compactly supported, because $\Log_{\PP_N}$ is uniformly bounded. 
Indeed, by compactness of the finite $\Omega$, any $T: \Omega\rightarrow \Omega$ is uniformly bounded.

Under such compacity, MK quantiles $\QQ_n$ and ranks $\RR_n$ benefit from uniform consistency on compact subsets \cite{chernozhukov2015mongekantorovich,Hallin_AOS_2021}. 
Based on these existing results, the next theorem describes the convergence of MK quantiles and ranks in LOT space in terms of Wasserstein distance in the image space.

\begin{theorem}\label{thm1}
Suppose that $\nu$ is a continuous distribution. 
For any $u\in\mathbb{B}(0,1)$, denote by $\overline{\QQ}_n(u) = \Exp_{\PP_N}\big(\QQ_n(u)\big)$ and $
\overline{\QQ}(u) = \Exp_{\PP_N}\big(\QQ(u)\big) $.
Also, for any image $I$, let
$\overline{\RR}_n(I) =  \RR_n\big(\Log_{\PP_N}(I) \big)$ 
and $\overline{\RR}(I) =  \RR\big(\Log_{\PP_N}(I) \big)$.
Then, for any compact $K\subset \mathbb{B}(0,1)$, it holds that
    \begin{equation*}
        \lim_{n\rightarrow 0} \sup_{u \in K} W_2^2\big(\overline{\QQ}_n(u),\overline{\QQ}(u)\big) = 0.
    \end{equation*}
   Moreover, for any $K'\subset\mathcal{P}_2(\Omega)$ such that $\Log_{\PP_N}(K')$ is compact in $\cT_{\PP_N}$,
    \begin{equation*}
        \lim_{n\rightarrow 0} \sup_{I \in K'} \Vert \overline{\RR}_n(I) - \overline{\RR}(I) \Vert_{\R^{2p}} = 0,
    \end{equation*}
    and 
    \begin{equation*}
        \lim_{n\rightarrow 0} \sup_{I \in K'} \vert D^{i}(I,{\PP_n}) - D^{i}(I,{\PP}) \vert = 0.
    \end{equation*}
\end{theorem}

\begin{proof}
The proof relies on the following well-known inequality, that is a direct byproduct of optimality in \eqref{MongeOT} combined with the definition of \eqref{Logmap},
$$
\W_2^2 \big( I,J \big)  \leq  \sum_{k=1}^p \Vert \Log_{\PP_N} (I)(\omega_k) - \Log_{\PP_N}(J)(\omega_k) \Vert^2 I_r(\omega_k).
$$
In our approach, the maps $\Log_{\PP_N} (I):\Omega \rightarrow \Omega$ are identified with the vectors $\big(\Log_{\PP_N} (I)(\omega_k)\big)_{k=1}^p \in \R^{2p}$.
Applying the above inequality, for any $u\in\mathbb{B}(0,1)$, 
\begin{equation}
        \W_2^2\left( \Exp_{\PP_N}\big(\QQ_n(u)\big), \Exp_{\PP_N}\big(\QQ(u)\big) \right) \leq \Vert\QQ_n(u) - \QQ(u) \Vert^2_{\R^{2p}}.
\end{equation}
    In addition, the right-hand size uniformly converges on any compact $K\subset \mathbb{B}(0,1)$ \cite{chernozhukov2015mongekantorovich}[Theorem 3.1], which yields the result. 
    Theorem 3.1 and Corollary 3.1 from \cite{chernozhukov2015mongekantorovich} also imply consistency for the ranks and the inner depth.

\end{proof}


\section{Numerical experiments}\label{NumExp}

\subsection{Datasets and implementation details}

LOT modeling implicitly assumes that a set of images is locally generated by a mass-preserving transformation from a template image \cite{park2018representing}. 
In accordance with these assumptions, we consider the following datasets of \textit{structural} images.  
The MNIST data set \cite{lecun1998gradient} is made of handwritten digits of size $28\times28$, with labels $k\in\{0,1,\cdots,9\}$. 
Two cell data sets \cite{basu2014detecting} are available within the PyTransKit package \cite{PyTransKit}, both with normal / cancer labels: $(i)$ The cell data set is made of $90$ images of size $64 \times 64$ and $(ii)$ the liver nuclear data set contains $20$ images of size $192\times 192$.
The chest X-rays dataset \cite{zunair2021synthesis} contains radiography images of size $256\times256$.
We only consider images labeled without COVID-19. 
To speed up LOT embedding, the X-ray images were reshaped to $100\times100$ pixels.  

Our procedure requires to solve several OT problems. 
LOT embeddings are learned with the simplex algorithm \cite{peyre2019computational}. 
To enforce continuity of the MK quantile function and of the inner depth, we use entropic regularization with a decreasing regularization parameter, with the implementation proposed in  \cite{kassraie2024progressive}. 

\subsection{Descriptive data analysis}

Quantiles and order statistics are among the most important tools in exploratory data analysis. 
With our concepts, a point $x_1$ is said to be deeper (or more central) than $x_2$ if
$
    D(x_1,{\PP_n}) \geq D(x_2,{\PP_n}),
$
where one can consider both the inner and the outer depth. 
For the analysis of a dataset $I_1,\cdots,I_n$ of images from the same distribution, the inner depth is appropriate. 
Then, order statistics can be defined by relabeling images so that 
\begin{equation}\label{ordstats}
     D^i \big( I_{(1)},{\PP_n} \big)  
    \geq  D^i \big( I_{(2)},{\PP_n} \big) 
    \geq \cdots \geq  D^i \big( I_{(n)},{\PP_n} \big).
\end{equation}
With this at hand, one can render the most extremes (resp. central) images within a distribution, as illustrated in Figures \ref{Fig_extremes_digits_intro} and \ref{Fig_extremes_Xray}.
In Figure 1, digits $2$ and $3$ are treated separately, each with $n=100$ samples. 
For each of these two distributions, the $10$ most central/ outward images with respect to the inner depth are represented. 
It illustrates that the center-outward ordering \eqref{ordstats} is consistent with intuition:
outward images have uncommon shape, as opposed to central ones.   
Figure \ref{Fig_extremes_Xray} shows results of the same experiment pursued on $n=265$ images from the chest X-ray dataset.
In Figure \ref{Fig_extremes_Xray}, one can observe that typical chest X-rays are better centered than atypical ones. 

\begin{figure}[h]
\centering
\includegraphics[width=\textwidth]{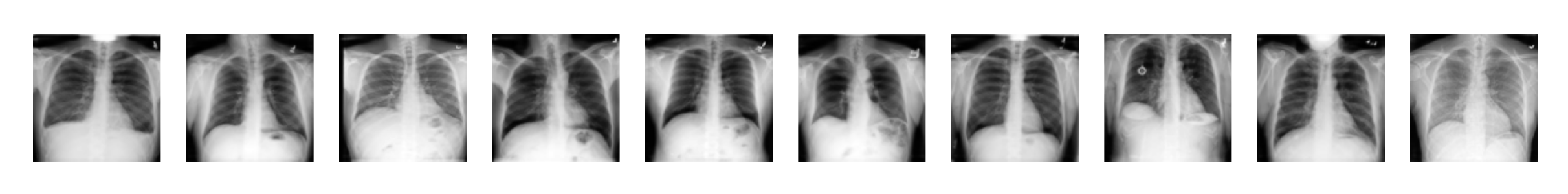}
\includegraphics[width=\textwidth]{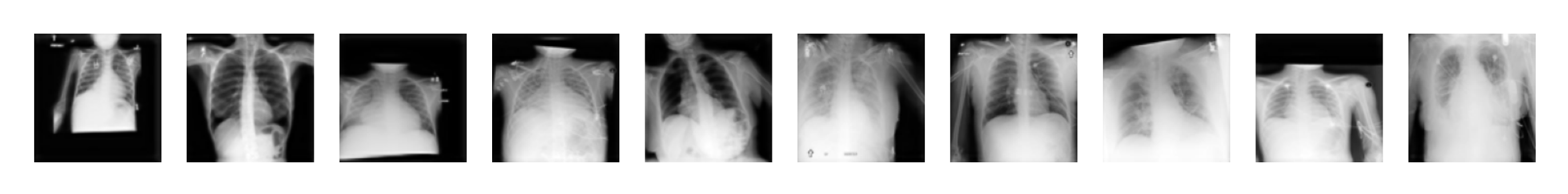}
\caption{Centrals (first row) and extremes (second row) for chest X-ray data.}
\label{Fig_extremes_Xray}
\end{figure}

Going further, note that a proportion $\alpha$ of the $n$ order statistics in \eqref{ordstats} is deeper than $I_{(\lceil \alpha n\rceil)}$.
As a byproduct, a summary of a dataset is provided by
\begin{equation*}
\{I_{(1)},I_{(\lceil 0.25 n\rceil)},I_{(\lceil 0.5 n\rceil)},I_{(\lceil 0.75 n\rceil)}, I_{(n)}\}.
\end{equation*}
This is exemplified in Figure \ref{Fig_summary}, and we stress that this mimics the traditional summary in dimension $d=1$, that includes the minimum, the quartiles, and the maximum. 
We also show depth values that range from $1/2$ for deepest points to $0$ for outward samples. 
For digits 2 and 3, $n=100$ samples were used, against $n=40$ for each class of cell data. 
This type of visualization provides an overview of both common and uncommon images within one class. 
For the cell dataset, one can observe that the major variation between classes corresponds to size and shape.

\begin{figure}[h]
\centering
\includegraphics[width=0.48\textwidth]{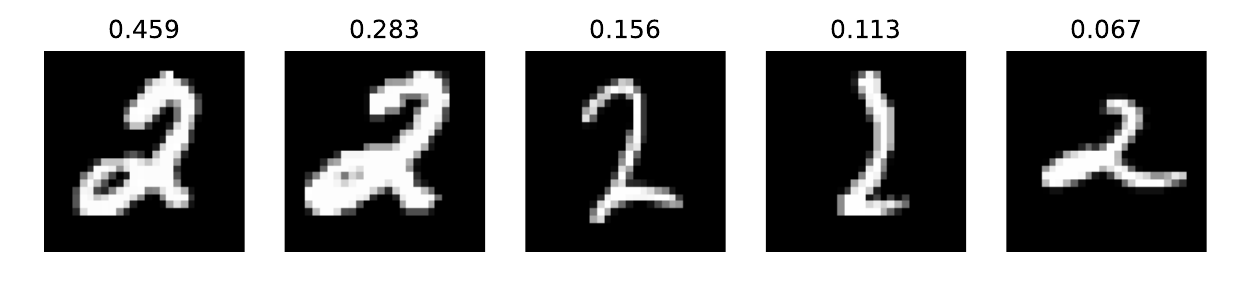}
\hspace{0.2cm}
\includegraphics[width=0.48\textwidth]{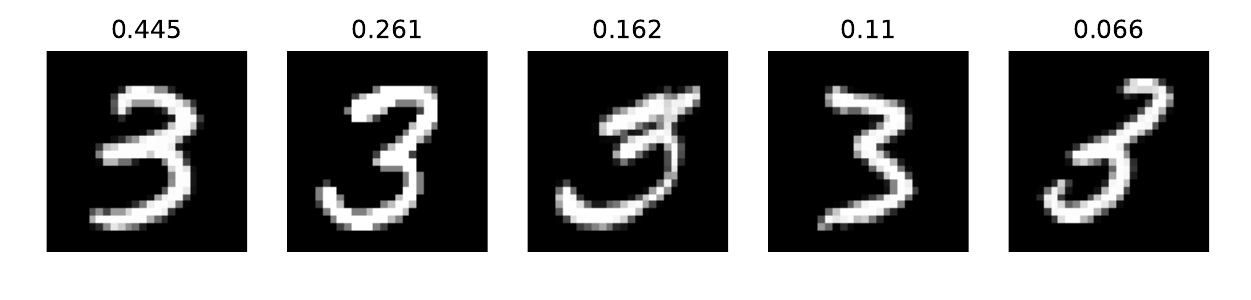}

\includegraphics[width=0.48\textwidth]{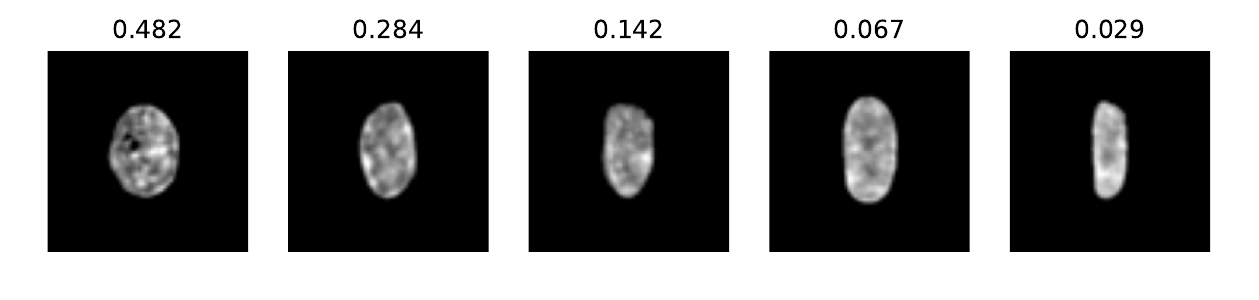}
\hspace{0.2cm}
\includegraphics[width=0.48\textwidth]{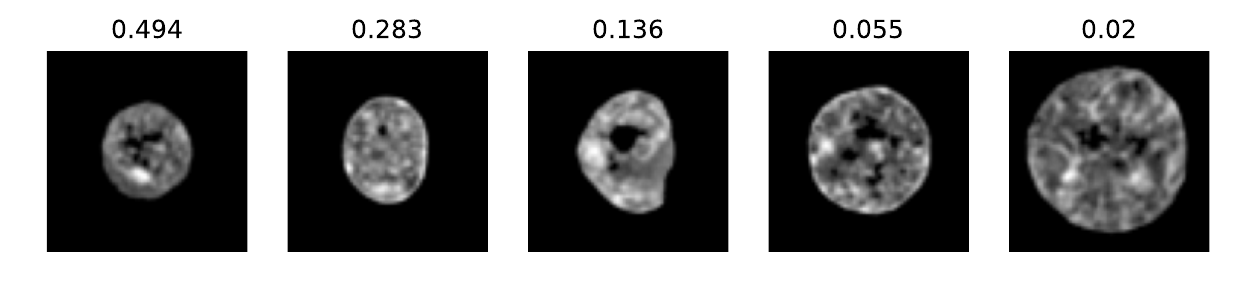}

\caption{
Summary statistics, for MNIST data (first row) and cell data (second row). 
For each data, two classes are treated as different distributions.}
\label{Fig_summary}
\end{figure}

\subsection{Testing equality of distributions}

A main virtue of MK quantiles is the distribution freeness of the associated ranks \cite{deb2023multivariate,Hallin_AOS_2021} that
is particularly appealing for statistical testing. 
Here, we extend for image data the two-sample Hotelling-Type test from \cite{deb2021pitman} that aims to answer
\begin{equation*}
    H_0 : \mu_1 = \mu_2 
    \hspace{1cm} \mbox{versus} \hspace{1cm}
    H_1 : \mu_1 \ne \mu_2.
\end{equation*}
This procedure departs from $m$ images $i.i.d.$ from $\mu_1$ and $n$ images $i.i.d.$ from $\mu_2$ that are considered as from the same distribution, under $H_0$. 
Hence, these are embedded in the Log-PCA space by taking the average in \eqref{AvgImage} within these $m+n$ images.
This yields 
samples $X_1,\cdots, X_m$ and $Y_1,\cdots,Y_n$ in $\R^d$ 
that form the empirical distribution
$
    \widehat{\mu}_{m+n} = \frac{1}{m+n} \sum_{i=1}^n \delta_{X_i} + \frac{1}{m+n} \sum_{j=1}^m \delta_{Y_j}.
$
In the Log-PCA space, we consider an isotropic Gaussian to be the reference distribution $\mu$, as it is beneficial for testing \cite{deb2021pitman}. 
We sample $m+n$ data points from $\mu$, and solve OT between this reference sample and $\widehat{\mu}_{m+n}$. 
Denoting by $\RR_{m,n}$ the induced MK rank map,
the test statistic is defined by 
\begin{equation*}
    T_{m,n} = \frac{mn}{m+n} \left\Vert \frac{1}{m} \sum_{i=1}^m \RR_{m,n}(X_i) - \frac{1}{n}\sum_{j=1}^n \RR_{m,n}(Y_j) \right\Vert^2.
\end{equation*}
From \cite{deb2021pitman}[Theorem 3.1], $T_{m,n}$ is asymptotically distributed according to $\chi_d^2$. 
Thus, denoting by $q_{1-\alpha}$ the $(1-\alpha)$-th quantile of the $\chi_d^2$ distribution, we reject $H_0$ if $T_{m,n} \geq q_{1-\alpha}$.
This results in a test that asymptotically controls the Type I error,
and that is consistent against a large class of alternative, \cite{deb2021pitman}.

\noindent Our numerical experiments estimate rejection rates $\mathbb{E}(\mathrm{1}_{\{T_{m,n} \geq q_{1-\alpha} \}})$ by taking the expectation over $N=50$ repetitions. 
Figure \ref{RejecRates} deals with low sample values: $m=n=50$ for MNIST, $m=n= 30$ for cell data and $m=n=10$ for liver nuclei data.  
All pairwise labels are considered for $\mu_1$ and $\mu_2$. 
One can observe the appropriate control under the null and the high power, except when comparing distributions of digits $5$ and $8$ and for the liver-nuclei data, when there are not enough observations.  

  \begin{figure}
    \centering
    \begin{tabular}{cc}
    \adjustbox{valign=b}{\subfloat[MNIST \label{subfig-1}]{%
          \includegraphics[width=5.2cm,height=4.6cm]{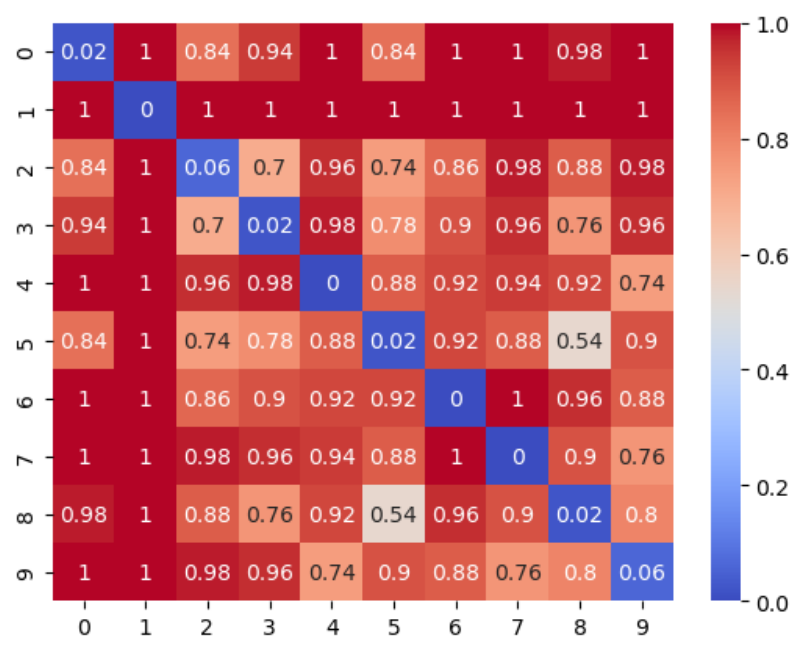}}}
    & 
    \adjustbox{valign=b}{\begin{tabular}{@{}c@{}}
    \subfloat[Cell data \label{subfig-2}]{%
          \includegraphics[width=0.2\linewidth,height=2cm]{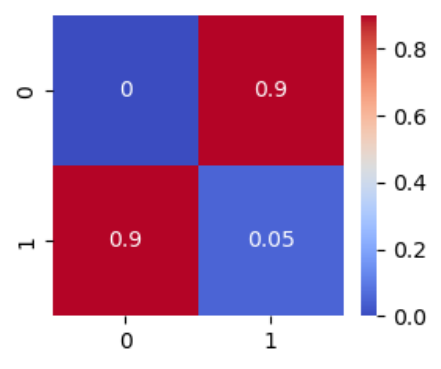}} \\
    \subfloat[Liver nuclei \label{subfig-3}]{%
          \includegraphics[width=0.2\linewidth,height=2cm]{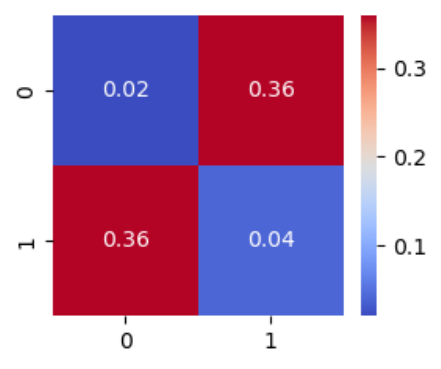}}
    \end{tabular}}
    \end{tabular}
    \caption{Rejection rates for Goodness-of-Fit testing for all pairwise labels. 
    }\label{RejecRates}
  \end{figure}

\subsection{Outlier Detection}

Outlier detection is an important application of statistical depths.
We depart from an unsupervised setting with one dataset of inliers, or correct images. 
After that we aim to discriminate new images between inliers and outliers. 

The procedure is as follows. 
We first embed the training data to learn MK quantiles. 
Then, to prevent from overfitting, we consider a calibration dataset, different from training instances. 
On these calibration data, we compute the order statistics \eqref{ordstats} for the inner and outer depth. 
For each depth, for $\alpha = 5\%$, an outlier will be an image with depth lower than the one of $I_{\lceil \alpha n\rceil}$. 
Figure \ref{DDplot} illustrates our methodology with a scatterplot of the outer depth as a function of the inner depth, for test data. 
We exemplify the two tresholds for the decision rule with dashed lines. 
One can observe that the inner and outer depth complement each other. 
On the left side of Figure \ref{DDplot}, the inliers correspond to digits $0$ from the MNIST dataset and outliers are taken from all other digits.  
On the right side, we added a few outliers to the chest X-ray data, that are radiography images of other organs. 
These outliers, detected by our procedure, are shown in Figure \ref{TruePositiveXray}. 

\begin{figure}[h]
\centering
\includegraphics[width=0.45\textwidth]{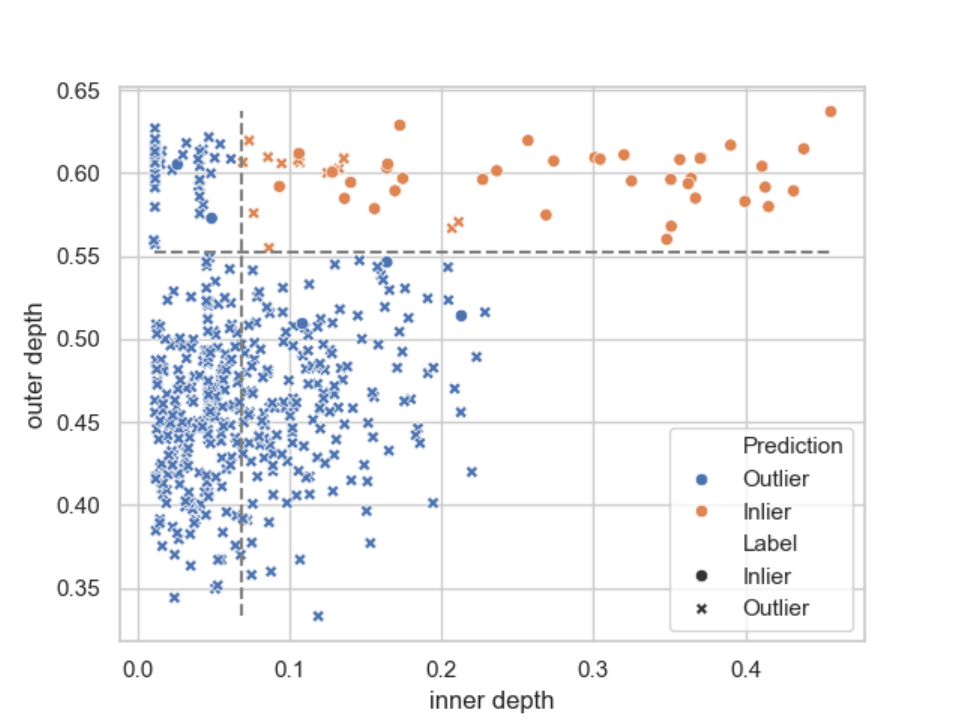}
\includegraphics[width=0.45\textwidth]{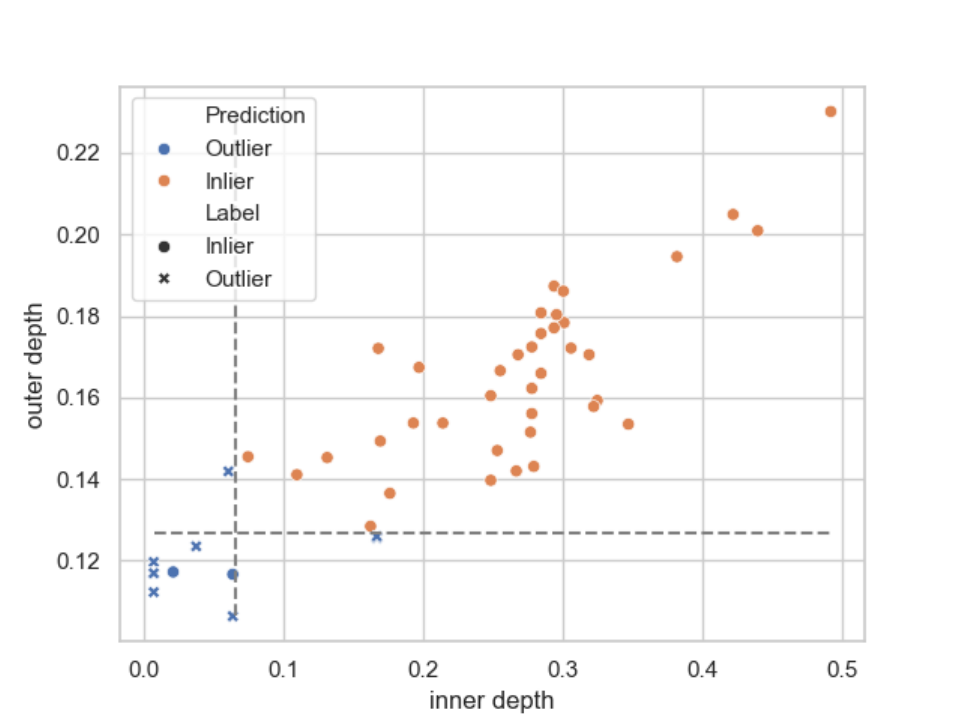}

\caption{Inner-outer DD plot for outlier detection. Dashed lines represent the tresholds to detect outliers. Left: MNIST data. Right: X-ray data.}
\label{DDplot}
\end{figure}

\begin{figure}[h]
\centering
\includegraphics[width=0.1\textwidth,height=0.1\textwidth]{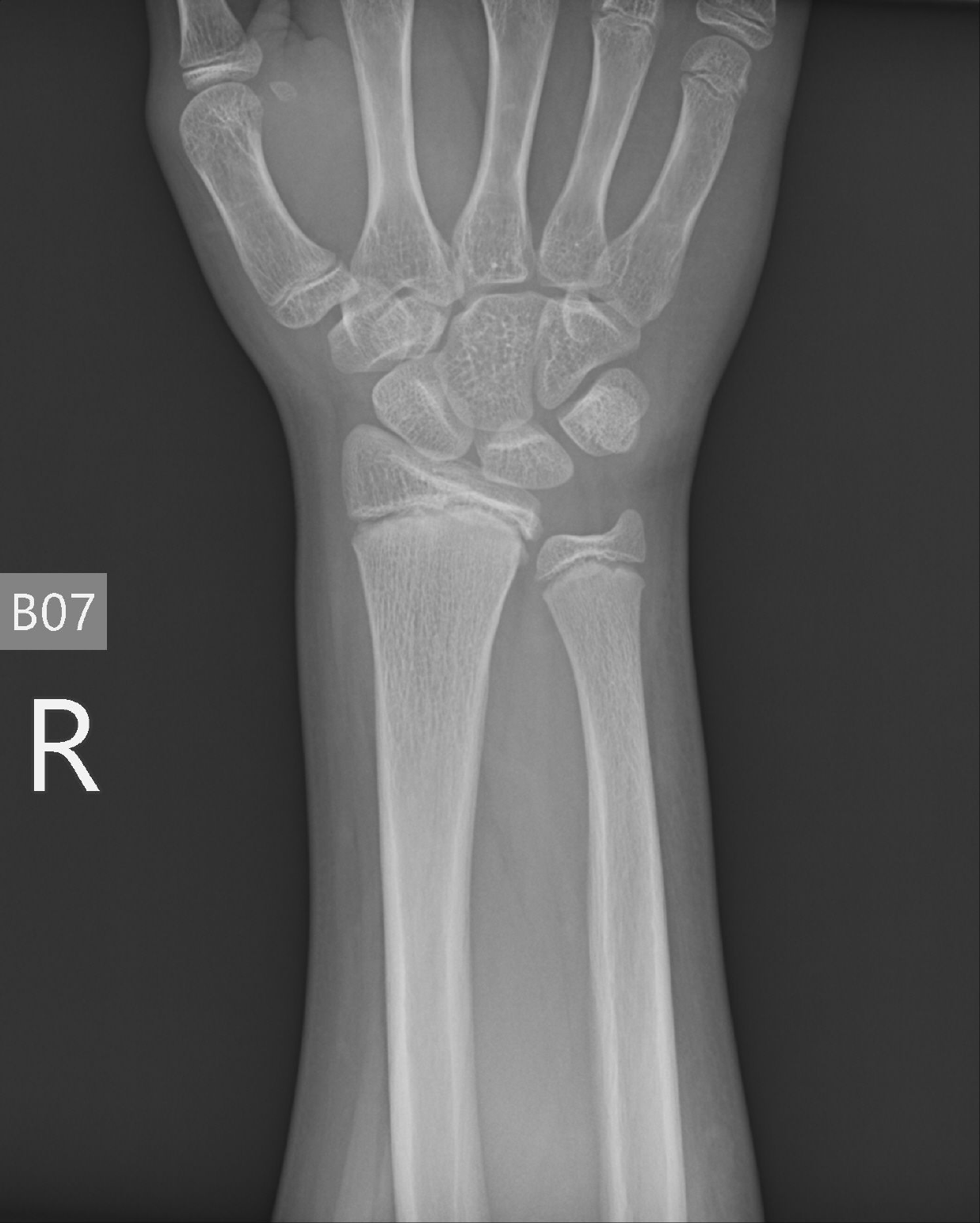}
\includegraphics[width=0.1\textwidth,height=0.1\textwidth]{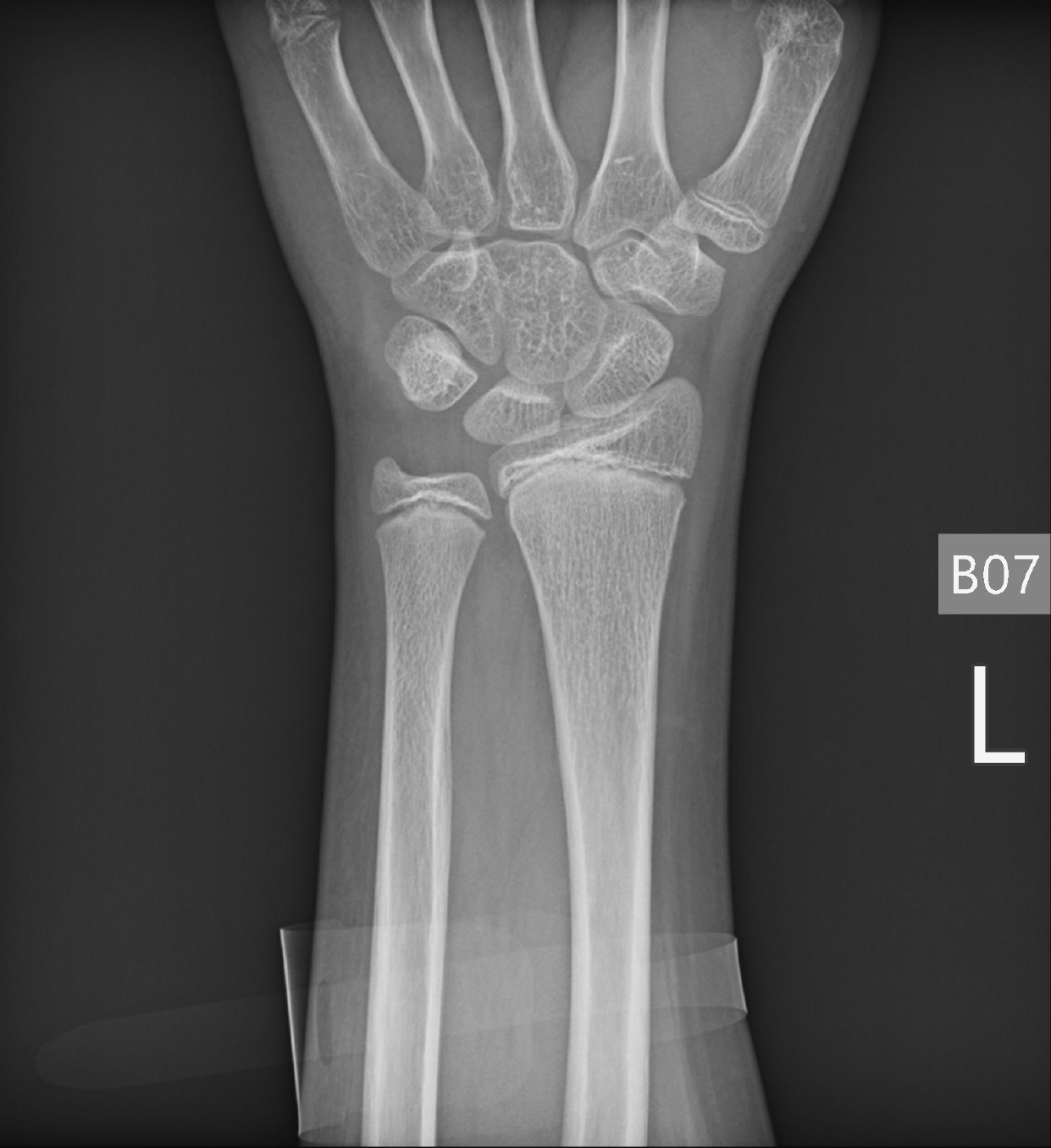}
\includegraphics[width=0.1\textwidth,height=0.1\textwidth]{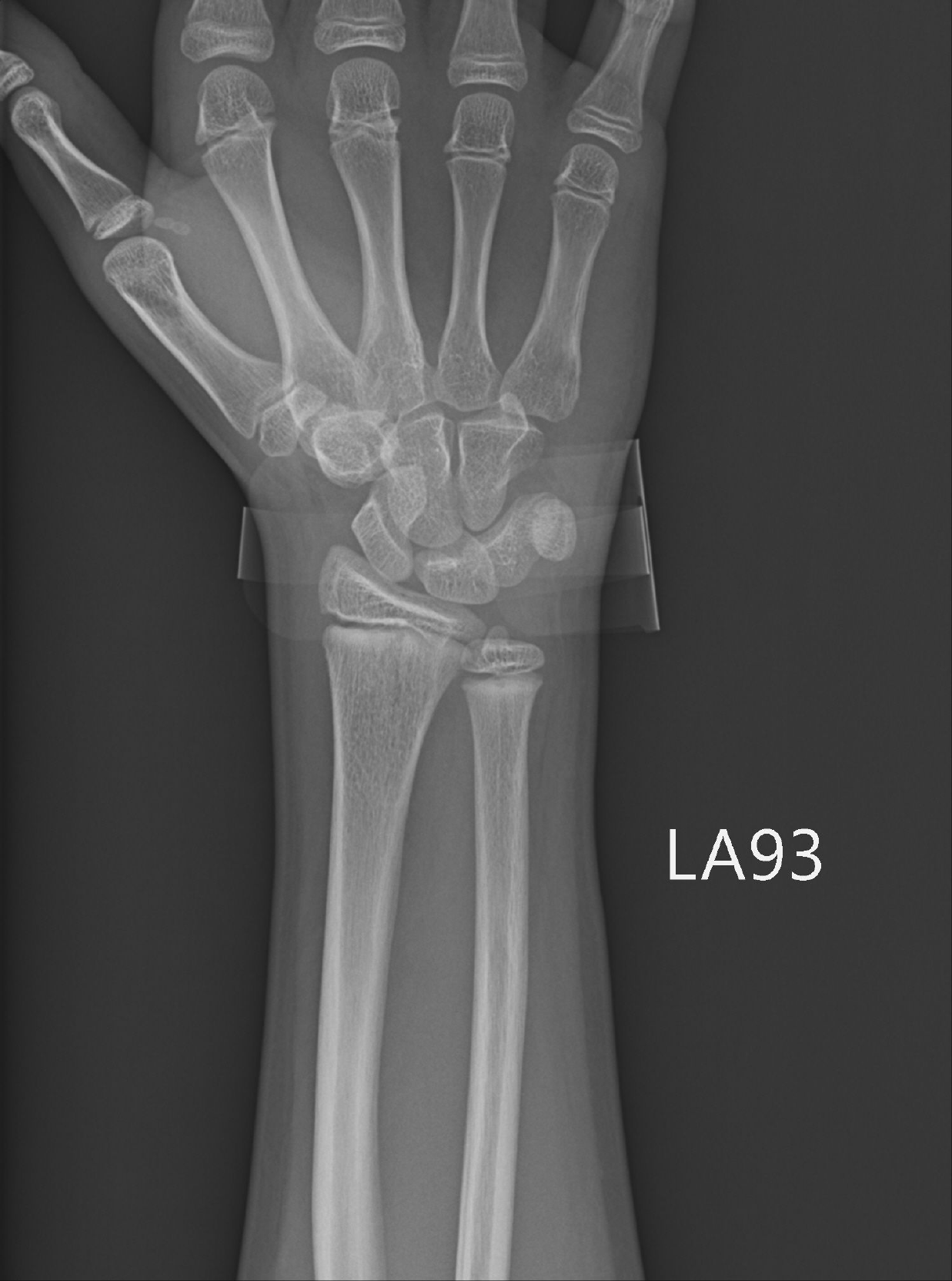}
\includegraphics[width=0.1\textwidth,height=0.1\textwidth]{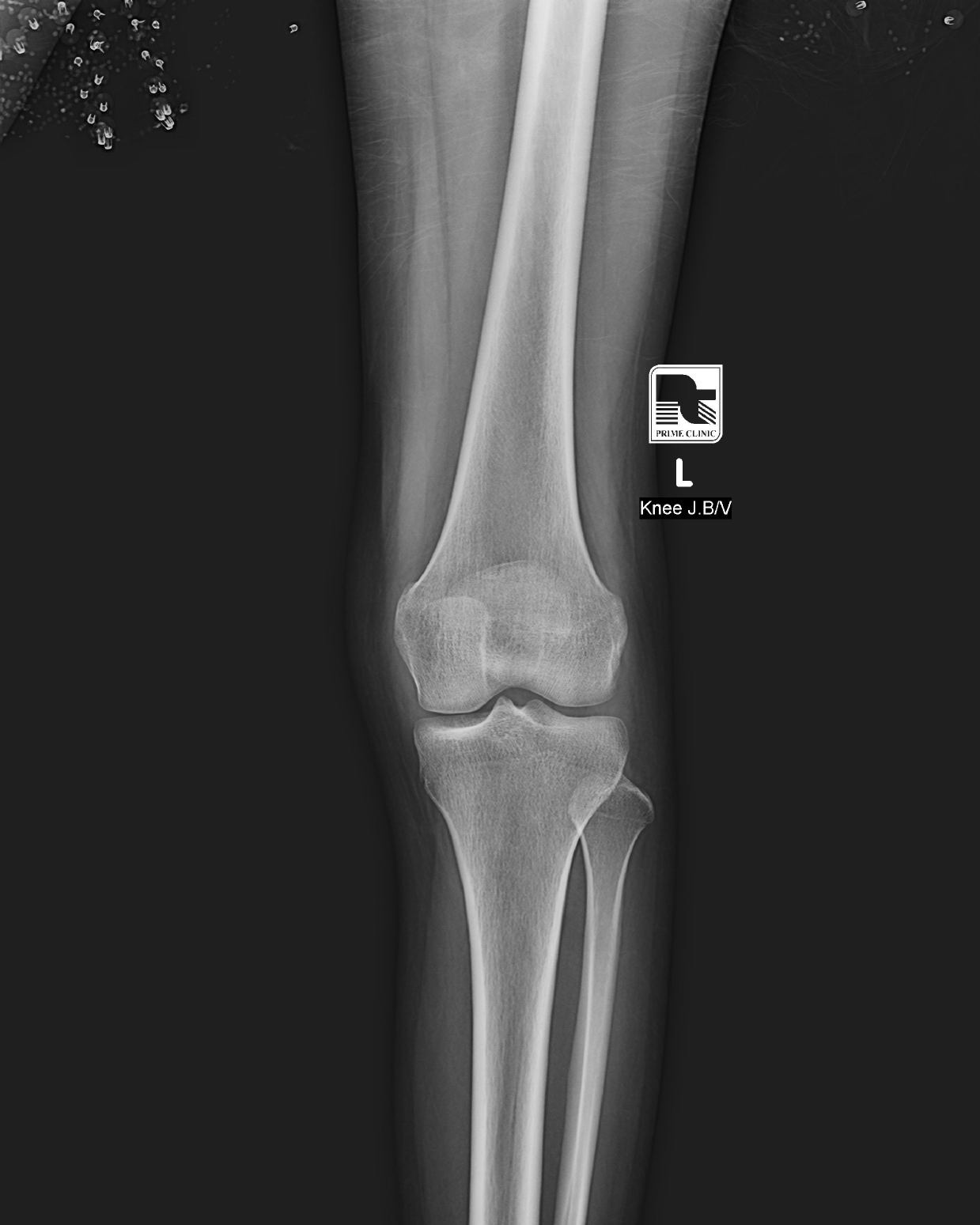}
\includegraphics[width=0.1\textwidth,height=0.1\textwidth]{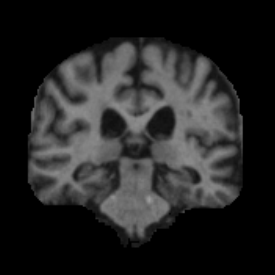}
\includegraphics[width=0.1\textwidth,height=0.1\textwidth]{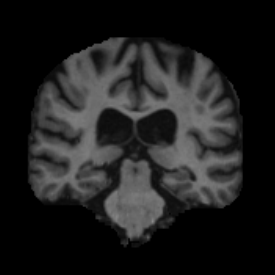}
\includegraphics[width=0.1\textwidth,height=0.1\textwidth]{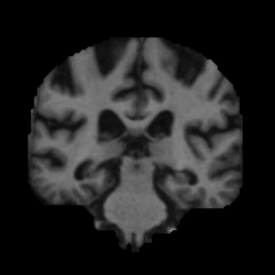}
\includegraphics[width=0.1\textwidth,height=0.1\textwidth]{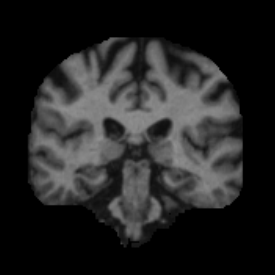}

\caption{Detected outliers for the chest X-ray data.}
\label{TruePositiveXray}
\end{figure}

In Figure \ref{Fig_ROCcurves}, this procedure is evaluated with ROC curves when varying $\alpha$ together with AUC measurements.  
We also considered the liver nuclei dataset, with outliers being images with the cancerous label. 
Note that, for this dataset, there is not enough data to separate between training and calibration, yet we obtain good performances.

  \begin{figure}
    \centering
    \begin{tabular}{ccc}
    \adjustbox{valign=b}{\subfloat[MNIST \label{ROCsubfig-1}]{%
          \includegraphics[width=0.25\textwidth]{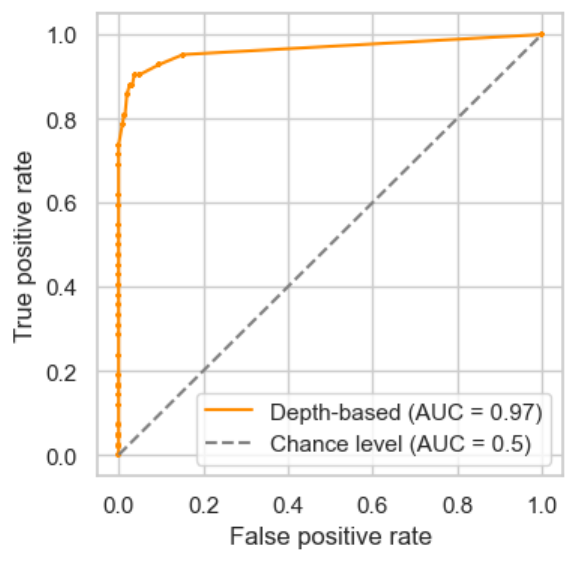}}}
    &   
    \adjustbox{valign=b}{\subfloat[X-ray \label{ROCsubfig-2}]{%
          \includegraphics[width=0.25\textwidth]{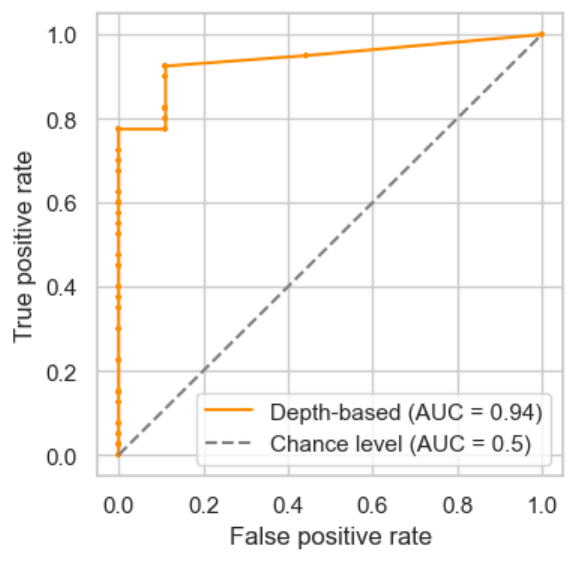}}}
    &
     \adjustbox{valign=b}{\subfloat[Liver nuclei \label{ROCsubfig-3}]{%
          \includegraphics[width=0.25\textwidth]{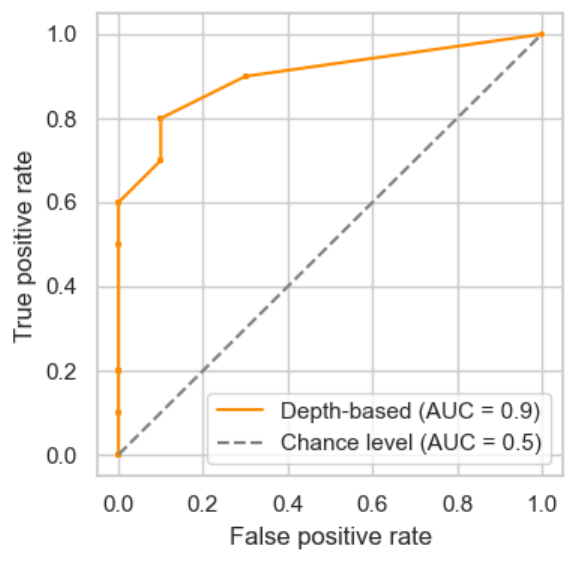}}}
    \end{tabular}
    \caption{ROC curves for outlier detection in three situations.}
    \label{Fig_ROCcurves}
  \end{figure}

\section{Conclusion}

This paper proposes to define MK quantiles and ranks for images by leveraging Linear Optimal Transport.  
Uniform convergence of empirical quantiles holds in the tangent space, as a direct consequence of existing consistency results. 
In our LOT framework, this convergence takes an interesting form due to the relation between LOT embeddings and the Wasserstein distance in image space. 
Finally, our numerical experiments illustrate the potential of this approach for data analysis.
Despite appealing properties, the reconstruction of images from the LOT space is limited in our current implementation. 
It would be interesting to explore whether the use of more involved algorithms \cite{peyre2019computational} than the simplex would yield better representations.

\

\textbf{Acknowledgements:} The author gratefully acknowledges financial support from the Agence Nationale de la Recherche (MaSDOL grant ANR-19-CE23-0017).

%
%

\bibliographystyle{splncs04}
\bibliography{biblioRankingImg}

\end{document}